\documentclass[11pt, english]{article}
\usepackage{graphicx}
\usepackage{verbatim}
\usepackage{amsmath}
\usepackage{amssymb}
\usepackage{amsthm}
\usepackage{color}
\usepackage{float}
\usepackage{algorithm,algorithmic}

\newtheorem{theorem}{Theorem}
\newtheorem{lemma}{Lemma}

\newcommand{\cardin}[1]{\left| {#1} \right|}

\newcommand{\B}{\mathcal{B}}

\renewcommand{\Re}{\mathbb{R}}

\renewcommand{\bar}[1]{\textstyle\overline{#1}}

\begin{document}

\title{Colorful Strips}
\author{Greg Aloupis
\thanks{Universit\'e Libre de Bruxelles, Brussels, Belgium. {\tt\{galoupis, jcardin, secollet, mkormanc, slanger, ptaslaki\}@ulb.ac.be}. Partially supported by the Communaut\'e fran\c caise de Belgique~-~ARC.} 
\and 
Jean Cardinal\footnotemark[1] 
\and 
S\'ebastien Collette\footnotemark[1]~~\thanks{Charg\'e de Recherches du FRS-FNRS.}
 \and 
Shinji Imahori\thanks{Nagoya University, Nagoya, Japan
{\tt imahori@na.cse.nagoya-u.ac.jp}}
\and 
Matias Korman\footnotemark[1] 
\and
Stefan Langerman\footnotemark[1]~~\thanks{Ma\^itre de Recherches du FRS-FNRS.}
\and
Oded Schwartz\thanks{The Weizmann Institute of Science, Rehovot, Israel, 
{\tt oded.schwartz@weizmann.ac.il}}
\and 
Shakhar Smorodinsky\thanks{Ben-Gurion University, Be'er Sheva, Israel.
{\tt shakhar@math.bgu.ac.il}}
\and 
Perouz Taslakian\footnotemark[1]
}
\maketitle

\begin{abstract}We study the following geometric hypergraph coloring problem:
given a planar point set and an integer $k$, we wish to color the points with $k$ colors so that any axis-aligned strip containing sufficiently many points contains all colors.  We show that if the strip contains at least $2k{-}1$ points, such a coloring can always be found. In dimension $d$, we show that the same holds provided the strip contains at least $k(4\ln k +\ln d)$ points. We also consider the dual problem of coloring a given set of axis-aligned strips so that any sufficiently covered point in the plane is covered by $k$ colors. We show that in dimension $d$ the required coverage is at most $d(k{-}1)+1$. 
This complements recent impossibility results on decomposition of strip coverings with arbitrary orientations.

From the computational point of view, we show that deciding whether a three-dimensional point set can be 2-colored so that any strip containing at least three points contains both colors is NP-complete. This shows a big contrast with the planar case, for which this decision problem is easy.

\end{abstract}

\section{Introduction}

There is a currently renewed interest in coloring problems on {\em geometric} hypergraphs, that is,
set systems defined by geometric objects. This interest is motivated by applications to wireless and
sensor networks~\cite{othersensors}; conflict-free colorings~\cite{shakharcf}, chromatic numbers~\cite{Sm07},
covering decompositions~\cite{pachtoth,Aloup2}, or polychromatic ({\em colorful}) colorings
of geometric hypergraphs~\cite{Aloup1} have been extensively studied in this context.

In this paper, we are interested in $k$-coloring finite point sets in $\Re^d$ so that any region bounded by two parallel axis-aligned hyperplanes, that contains at least
some fixed number of points, also contains a point of each color.

An {\em axis-aligned strip}\footnote{From here on, unless otherwise specified, a {\em strip} is always assumed to be axis-aligned} is the area enclosed between two parallel axis-aligned hyperplanes. A {\em $k$-coloring} of a finite set assigns one of $k$ colors to each element in the set.
Let $S$ be a $k$-colored set of points in $\Re^d$. A strip is said to be {\em polychromatic}
with respect to $S$ if it contains at least one element of each color class.
We define the function $p(k,d)$ as the minimum number for which there always exists
a $k$-coloring of any point set in $\Re^d$ such that every strip containing at least $p(k,d)$ points is polychromatic.
This is a particular case of the general framework proposed by Aloupis,
Cardinal, Collette, Langerman, and Smorodinsky in~\cite{Aloup1}.

Note that the problem does not depend on whether the strips are open or closed, since the problem can be seen in a purely combinatorial fashion: an axis-aligned strip isolates a subsequence of the points in sorted order with respect to one of the axes.  
Therefore, the only thing that matters is the order in which the points appear along each axis.
We can thus rephrase our problem, considering $d$-dimensional points sets, as finding the minimum
value $p(k,d)$ such that the following holds:
For $d$ permutations of a set of items $S$, it is always possible to color the items with $k$ colors, so that in all $d$
permutations every sequence of at least $p(k,d)$ contiguous items contains one item of each color.\medskip

We also study {\em circular} permutations, in which the first and the last elements are contiguous.
We consider the problem of finding a minimum value $p'(k,d)$ such that, for any $d$
circular permutations of a set of items $S$, it is possible to $k$-color the items so that in
every permutation, every sequence of $p'(k,d)$ contiguous items contains all colors.

A restricted geometric version of this problem in $\Re^2$ consists of coloring a point set $S$ with respect to wedges.
For our purposes, a wedge is any area delimited by two half-lines with common endpoint at one of $d$ given apices.
Each apex induces a circular ordering of the points in $S$. This is illustrated in Figure~\ref{fig}.
We aim at coloring $S$ so that any wedge  containing at least $p'(k,d)$ points is polychromatic.
In $\Re^2$, the non-circular case corresponds to wedges with apices at infinity, hence the circular case can be seen as a generalization.

\begin{figure}
\center\includegraphics[width=\textwidth]{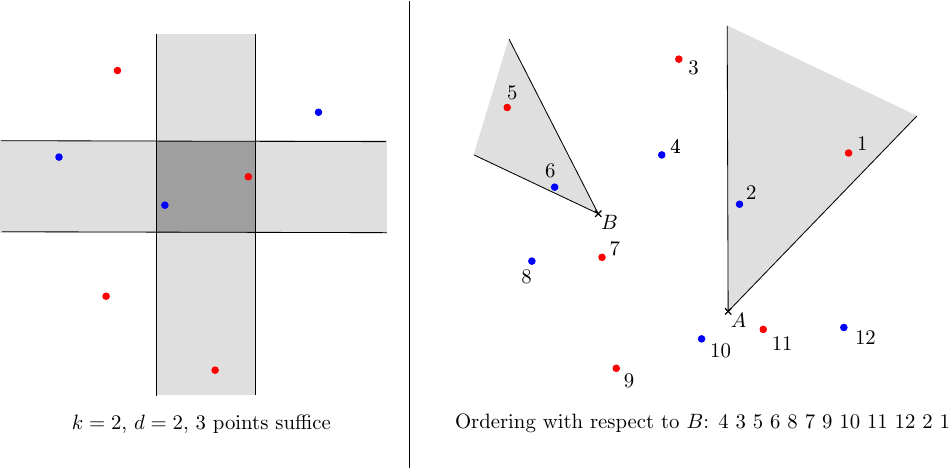}
\caption{\label{fig} Illustration of the definitions of $p(k,2)$ and $p'(k,2)$.
On the left, points are 2-colored so that any axis-aligned strip containing at least three points is bichromatic.
On the right, two points $A$ and $B$ define two circular permutations of the point set.
In this case, we wish to color the points so that there is no long monochromatic subsequence in either of the two circular orderings.}
\end{figure}

We then study a dual version of the problem, in which a set of axis-aligned strips is to be colored so
that sufficiently covered points are contained in strips from all color classes.
For instance, in the planar case we study the following function $\bar{p}(k,d)$. Let $H$ be a $k$-colored set of strips in $\Re^d$.
A point is said to be polychromatic with respect to $H$ if it is contained in strips of all
$k$ color classes. The function $\bar{p}(k,d)$ is the minimum number for which there always exists
a $k$-coloring of any set of strips in $\Re^d$ such that every point of $\Re^d$ contained in at least $\bar{p}(k,d)$ strips is polychromatic.


Note that the functions $p(k,d)$, $p'(k,d)$ and $\bar{p}(k,d)$ are monotone and non-decreasing (in particular, they all go to infinity when either $k$ or $d$ goes to infinity). Since we are interested in arbitrarily large pointsets, we always consider the set that we color to be ``large enough'' (that is, unbounded in terms of $k$).

\paragraph{Previous results.}
A {\em hypergraph} $(S,R)$ is defined by a set $S$ (called the {\em ground set}) and a set $R$ of
subsets of $S$. The main problem studied here is the coloring of geometric hypergraphs where the
ground set $S$ is a finite set of points, and the set of ranges $R$ consists of all subsets of $S$ that
can be isolated by a single strip. In the dual case the ground set $S$ is a finite set of geometric shapes
and the ranges are points contained in the common intersection of a subset of $S$.
In some places in the literature, finite geometric hypergraphs are also referred to as geometric {\em range spaces}.

Several similar problems have been studied in this context~\cite{Pa80,TT07,Aloup1}, where the range space is not defined by strips,
but rather by halfplanes, triangles, disks, pseudo-disks or translates of a centrally symmetric convex polygon.
The problem was originally stated in terms of decomposition of {\em $c$-covers} (or \emph{$f$-fold coverings}) in the plane:
A $c$-cover of the plane by a convex body $Q$ ensures that every point in the plane is covered by at least $c$ translated copies of $Q$.
In 1980, Pach~\cite{Pa80} asked if, given $Q$, there exists a function $f(Q)$ such that every $f(Q)$-cover of the plane
can be decomposed into $2$ disjoint $1$-covers. A natural extension is to ask if given $Q$, there exists a function $f(k,Q)$ such
that every $f(k,Q)$-cover of the plane can be decomposed into $k$ disjoint $1$-covers.
This corresponds to a $k$-coloring of the $f(k,Q)$-cover, such that every point of the plane is polychromatic.

Partial answers to this problem are known: Pach~\cite{PDMPC80} referenced an unpublished manuscript by Mani and Pach~\cite{decompball} showing that any 33-cover of the plane by unit disks can be decomposed into two $1$-covers. This would imply that the function $f$ exists for unit disks, but could still be exponential in $k$.
Recently, Tardos and T\'{o}th~\cite{TT07} proved that any 43-cover by translated copies of a triangle can be decomposed into two $1$-covers.
For the case of centrally symmetric convex polygons, Pach~\cite{Pach86} proved that $f$ is at most exponential in $k$.
More than 20 years later, Pach and T{\'o}th~\cite{pachtoth} improved this by showing that
$f(k,Q) = O(k^{2})$, and was afterwards Aloupis et al.~\cite{Aloup2} proved that $f(k,Q)=O(k)$. Recently, Gibson and Varadarajan~\cite{Gib-Vara09} showed that the same property also holds for any arbitrary convex polygon $Q$.

On the other hand, for the range space induced by arbitrary disks, Mani and Pach~\cite{decompball}~(see also~\cite{pachindecomp})
proved that $f(2,Q)$ is unbounded: for any constant $c$, there exists a set of points that cannot be $2$-colored so that all open disks
containing at least $c$ points are polychromatic. Pach, Tardos and T{\'o}th~\cite{pachindecomp} obtained a similar
result for the range spaces induced by the family of either non-axis-aligned strips, or axis-aligned rectangles.
Specifically, for any integer $c$ there exist $c$-fold coverings with non-aligned strips that cannot be decomposed into two
coverings (i.e., cannot be $2$-colored). The previous impossibilities constitute our main motivation for introducing the problem
of $k$-coloring axis-aligned strips, and strips with a bounded number of orientations.

\paragraph{Paper Organization.}
In Section~\ref{sec_plane} we give constructive upper bounds on the functions $p$ and $p'$ for $d=2$. In Section~\ref{higher_dim} we consider higher-dimensional cases, as well as the computational complexity of finding a valid coloring. Section~\ref{sec_dual} concerns the dual problem of coloring strips with respect to points. 
 Our lower and upper bounds  are summarized in Table~\ref{tbl_res}.

\begin{table}
\caption{Bounds on $p$, $p'$ and $\bar{p}$}\label{tbl_res}
\renewcommand{\arraystretch}{1.5}
\center\begin{tabular}{|c|c|c|c|}
\hline
 & $p(k,d)$ & $p'(k,d)$ & $\bar{p}(k,d)$  \\ \hline
upper bound  & $k(4\ln k +\ln d)$& $k(4\ln k +\ln d)$ & $d(k{-}1)+1$ \\
&($2k{-}1$ for $d{=}2$)&($2k$ for $d{=}2$) &  \\
\hline
lower bound  & \multicolumn{2}{c|}{$2\cdot \lceil \frac{(2d-1)k}{2d} \rceil -1$} & $ \lfloor k/2 \rfloor d+1$ \\ \hline
\end{tabular}
\end{table}

\section{Axis-aligned strips and circular permutations for $d=2$}
\label{sec_plane}

We first consider upper bounds for the functions $p(k,2)$ and $p'(k,2)$.

\subsection{Axis-aligned strips: Upper bound on $p(k,2)$}

We refer to a strip containing at least $i$ points as an {\em $i$-strip}.
Our goal is to show that for any integer $k$ there is a
constant $p(k,2)$ such that any finite planar point
set can be $k$-colored so that all $p(k,2)$-strips are polychromatic.\medskip

For $d=2$, there is a reduction to the recently studied
problem of $2$-coloring graphs so that monochromatic
components are small. Haxell et al.~\cite{HST} proved that the vertices of
any graph with maximum degree $4$ can be $2$-colored so
that every monochromatic connected component has size at most
$6$. For a given finite point set $S$ in the plane, let $E$ be the set
of all pairs of points $u,v \in S$ such that there is a strip containing
only $u$ and $v$. The graph $G=(S,E)$ has maximum degree $4$,
as it is the union of two paths. By the results of \cite{HST}, $G$ can be
$2$-colored so that every monochromatic connected component has
size at most $6$. In particular every path of size at
least $7$ contains points from both color classes. To finish the
reduction argument one may observe that every strip
containing at least $7$ points corresponds to a path (of size at
least $7$) in $G$.
We improve and generalize this first bound in the following.

\begin{theorem}
\label{poly-strips}
For any finite planar set $S$ and any integer $k$, $S$ can
be $k$-colored so that any $(2k{-}1)$-strip is polychromatic.
That is,
$$
p(k,2) \leq 2k-1.
$$
\end{theorem}
\begin{proof}
Let $s_1,\ldots,s_n$ be the points of $S$ sorted by (increasing) $x$-coordinates and let
$s_{\pi_1},\ldots,s_{\pi_n}$ be the sorting by $y$-coordinates. We first
assume that $k$ divides $n$, and later show how to remove the need for this assumption.
Let $V_x$ be the set of $n/k$ disjoint contiguous $k$-tuples in $s_1,\ldots,s_n$.
Namely, $V_x=\{\{s_1,\ldots,s_k\},\{s_{k+1},\ldots,s_{2k}\},\ldots,\{s_{n-k+1},\ldots,s_n\}\}$.
Similarly, let $V_y$ be the $k$-tuples defined by $s_{\pi_1},\ldots,s_{\pi_n}$.\medskip

We define a bipartite multigraph $G=(V_x,V_y,E)$ as follows:
For every pair of $k$-tuples $A\in V_x$, $B\in V_y$, we include an edge
$e_s =\{A,B\} \in E$ if there exists a point $s$ in both $A$ and $B$.
Note that an edge $\{A,B\}$ has multiplicity $\cardin{A\cap B}$ and that
the number of edges $\cardin{E}$ is $n$. The
multigraph $G$ is $k$-regular because every $k$-tuple $A$ contains
exactly $k$ points and every point $s\in A$ determines exactly one
incident edge labeled $e_s$. It is well known that the chromatic index of
any bipartite $k$-regular multigraph is $k$ (and can be
efficiently computed, see e.g., \cite{Alon03,ColeOstSchirra01}).
Namely, the edges of such a multigraph can be partitioned into
$k$ perfect matchings. Let $E_1,\ldots,E_k$ be such a
partition and $S_i \subset S$ be the set of labels of the edges
of $E_i$. The sets $S_1,\ldots,S_k$ form a partition (i.e., a
coloring) of $S$. We assign color $i$ to the points of $S_i$.\medskip

We claim that this coloring ensures that any
$(2k{-}1)$-strip is polychromatic. Let $h$ be a $(2k{-}1)$-strip and assume without
loss of generality that $h$ is parallel to the $y$-axis. Then $h$
contains at least one $k$-tuple $A \in V_x$. By the properties of
the above coloring, the edges incident to $A$ in $G$ are colored
with $k$ distinct colors. Thus, the points that correspond to the labels of
these edges are colored with $k$ distinct colors, and $h$ is
polychromatic.\medskip

To complete the proof, we must handle the case where $k$ does
not divide $n$. Let $i= n\pmod k$.
Let $Q=\{q_1,\ldots,q_{k-i}\}$
be an additional set of $k{-}i$ points, all located to the right and above
the points of $S$. We repeat our previous construction on $S\cup Q$.
Now, any $(2k{-}1)$-strip which
is, say, parallel to the $y$-axis will also contain a $k$-tuple
$A \in V_x$ disjoint from $Q$. Thus our arguments follow as before.
\end{proof}

The proof of Theorem~\ref{poly-strips} is constructive and leads directly to
an $O(n\log n)$-time algorithm to $k$-color $n$ points in the plane so that every
$(2k{-}1)$-strip is polychromatic.
The algorithm is simple: we sort $S$, construct $G=(V_x,V_y,E)$, and color the edges of $G$ with $k$ colors.
The time analysis is as follows: sorting takes $O(n \log n)$ time. Constructing $G$ takes $O(n+|E|)$ time.
As $G$ has $\frac{2n}{k}$ vertices and is $k$-regular, it has $n$ edges; so this step takes $O(n)$ time. Finding the
edge-coloring of $G$ takes $O(n \log n)$ time~\cite{Alon03}. The total running time is therefore $O(n \log n)$.

\subsection{Circular permutations: Upper bound on $p'(k,2)$}
\label{sec_wedges}

We now consider the value of $p'(k,d)$. Given $d$ circular permutations of a set $S$, we color $S$ so that every sufficiently long subsequence in any of the circular permutations is polychromatic.
The previous proof for $p(k,d)\leq 2k{-}1$ (Theorem~\ref{poly-strips}) does not hold when we consider circular permutations. However, a slight modification provides the same upper bound, up to a constant term.

\begin{theorem}\label{theo_pprime}
\hskip 0.3in $p'(k,2) \leq 2k$
\end{theorem}
\begin{proof}
If $k$ divides $n$, we separate each circular permutation into $n/k$ sets of size $k$. We define a multigraph, where the vertices represent the sets of $k$ items, and there is an edge between two vertices if two sets share the same item. Trivially, this graph is $k$-regular and bipartite, and can thus be edge-colored with $k$ colors. Each edge in this graph corresponds to one item in the permutation, thus each set of $k$ items contains points of all $k$ colors.\medskip

If $k$ does not divide $n$, let $a=\lfloor n/k \rfloor$, and $b=n\pmod k$.  If $a$ divides $b$, we separate each of the two circular permutations into $2a$ sets, of  alternating sizes $k$ and $b/a$. Otherwise, the even sets will also alternate between size $\lceil b/a \rceil$ and $\lfloor b/a \rfloor$, instead of $b/a$. We extend both permutations by adding dummy items to each set of size less than $k$, so that we finally have only sets of size $k$. Dummy items appear in the same order in both permutations. We can now define the multigraph just as before. \medskip

If we remove the dummy nodes, we deduce a coloring for our original set. As each color appears in every set of size $k$, the length of any subsequence between two items of the same color is at most $2(k-1) + \lceil b/a \rceil$. Therefore, $p'(k,2)\leq 2(k-1) + \lceil b/a \rceil +1$.\medskip

Finally, if $n \ge k(k-1)$, then $a \ge k-1$, and $b \le a$, we know that $\lceil b/a \rceil\le 1$, and thus $p'(k,2)\leq 2k$.
\end{proof}

\section{Higher dimensional strips}
\label{higher_dim}

In this section we study the same problem for strips in higher dimensions.
We provide upper and lower bounds on $p(k,d)$.
We then analyze the coloring problem from a computational viewpoint, and show that
deciding whether a given instance $S \subset \Re^d$ can be 2-colored
such that every 3-strip is polychromatic is NP-complete.

\subsection{Upper bound on strip size, $p(k,d)$}\label{uperB}

\begin{theorem}
\label{lll_strips}
Any finite set of points $S \subset \Re^d$ can be $k$-colored so that every axis-aligned strip
containing $k(4 \ln k +\ln d)$ points is polychromatic, that is,
$$
p(k,d) \leq k(4 \ln k +\ln d).
$$
\end{theorem}
\begin{proof}
The proof uses the probabilistic method. Let $\{1,\ldots,k\}$ denote the set of $k$ colors. We randomly
color every point in $S$ independently
so that a point $s$ gets color $i$ with probability
$\frac{1}{k}$ for $i=1,\ldots,k$. For a $t$-strip $h$, let $\B_h$ be the ``bad'' event where $h$
is not polychromatic. It is easily seen that
$\Pr [\mathcal{B}_h] \leq k (1-\frac{1}{k})^t$.
Moreover, $\B_h$ depends on at most $(d-1)t^2 + 2t-2$ other events.
Indeed, $\B_h$ depends only on $t$-strips that share points with $h$. Assume without
loss of generality that $h$ is orthogonal to the $x_1$ axis. Then $\B_h$ has a non-empty intersection with at most $2(t-1)$
other $t$-strips which are orthogonal to the $x_1$ axis. For each of the other $d{-}1$ axes, $h$ can intersect at most $t^2$
$t$-strips since every point in $h$ can belong to at most $t$ other $t$-strips.\medskip

By the Lov\'asz Local Lemma, (see, e.g., \cite{AS00}) we have that
if $t$ satisfies
$$
e \cdot \left((d-1)\cdot t^2+ 2t-1\right) \cdot k\left(1-\frac{1}{k}\right)^t < 1
$$
(where $e$ is the basis of the natural logarithm), then
$$
\Pr\left[\bigwedge_{\cardin{h} = t} \bar{\B_h}\right]>0.
$$
In particular, this means that there exists a $k$-coloring for which
every $t$-strip is polychromatic. It can be verified that $t=k(4 \ln k +\ln d)$ satisfies the condition.
\end{proof}

The proof of Theorem~\ref{lll_strips} is non-constructive. However, we can use known algorithmic versions of the Local Lemma (see for instance~\cite{1536462}) to obtain a constructive proof. 
Also note that Theorem~\ref{lll_strips} holds in the more general case where the strips are not necessarily axis-aligned. In fact, one can have a total of $d$ arbitrary strip orientations  in some fixed arbitrary dimension and the proof will hold verbatim. Finally, we note that the same proof also works for the case of circular permutations, yielding the same upper bound:

\begin{theorem}
\label{lll_circular}
\hskip 0.3in $p'(k,d) \leq k(4 \ln k +\ln d)$
\end{theorem}

\subsection{Lower bound on $p(k,d)$}

We first introduce a well-known result on the decomposition of complete graphs:
\begin{lemma}
\label{lem-decompEven}
The edges of $K_{2h}$ can be decomposed into $h$ pairwise edge-disjoint Hamiltonian paths.
\end{lemma}

This result follows from a special case of the Oberwolfach problem~\cite{Alspach08}.  An explicit proof of this lemma can also be found in \cite{decompEven}. 

Note that if the vertices of $K_{2h}$ are labeled  $V=\{1, \ldots, 2h\}$, each path can be seen as a permutation of $2h$ elements. Using Lemma~\ref{lem-decompEven} we obtain:

\begin{theorem}\label{thm_lb}
For any fixed dimension $d$ and number of colors $k$, let $s=\left\lceil\frac{(2d-1)k}{2d}\right\rceil -1$. Then,
$$
p'(k,d) \geq p(k,d) \geq 2s +1.
$$
\end{theorem}

\begin{proof}
The first inequality comes from the fact that any polychromatic coloring with respect to circular permutations is also polychromatic with respect to strips. We now focus on showing the second inequality: let $\sigma_1, \ldots, \sigma_d$ be any decomposition of $K_{2d}$ into $d$ paths: we construct the set $P=\{p_i| 0 \leq i \leq 2d \}$, where $p_i=(\sigma_1(i), \ldots, \sigma_d(i))$. Note that the ordering of $P$, when projected to the $i$-th axis, gives permutation $\sigma_i$. Since the elements $\sigma$ decompose $K_{2d}$, in particular for any $i,j \leq 2d$ there exists a permutation in which $i$ and $j$ are adjacent.

We replace each point $p_i$ by a set $A_i$ of $s$ points
arbitrarily close to $p_i$. By construction, for any $i,j \leq
2d$, there exists a $2s$-strip containing exactly $A_i \cup A_j$.
Consider any possible coloring of the sets $A_i$: since $|A_i|=s$
and we are using $k$ colors, there are at least $k-s$ colors not
present in any set $A_i$. 

Since $\left\lceil\frac{(2d-1)k}{2d}\right\rceil -1 < \frac{(2d-1)k}{2d}$, we conclude that $k-s > k- \frac{(2d-1)k}{2d}= k/2d$. That is, each set is missing strictly more than $k/2d$ colors. By the pigeonhole principle, there exist $i$ and $j$ such that the set $A_i \cup A_j$ is
missing a color (otherwise there would be more than $k$ colors).
In particular, the strip that contains set $A_i \cup A_j$ is not
polychromatic, thus the theorem is shown.

We gave a set of bounded size $n=2d$ reaching the lower bound, but we can easily create larger sets reaching the same bound: we can add as many dummy points as needed at the end of every permutation, which does not decrease the value of $p(k,d)$.
\end{proof}

Note that, assymptotitically speaking, the lower bound does not depend on $d$. However, by the negative result of \cite{pachindecomp}, we know that $p(k,d) \rightarrow \infty$ when $d\rightarrow \infty$.

\subsection{Computational complexity}
In Section~\ref{sec_plane}, we provided an algorithm that finds
a $k$-coloring such that every planar $(2k{-}1)$-strip is polychromatic.
Thus for $d{=}2$ and $k{=}2$, this yields a $2$-coloring such that every $3$-strip is polychromatic.

Note that in this case $p(2,2)=3$, but the minimum required size of a strip for a given instance can be either $2$ or $3$. Testing if it is equal to 2 is easy: we can simply alternate the colors in the first permutation, and check if they also alternate in the other. Hence the problem of minimizing the size of the largest monochromatic strip on a given instance is polynomial for $d=2$ and $k=2$. We now show that it becomes NP-hard for $d>2$ and $k=2$. The same problem for $k>2$ is left open.

\begin{theorem}\label{theoNPhard}
The following problem is $NP$-complete:\\
{\bf {\em Input}:} $3$ permutations $\pi _1, \pi _2,\pi _3 $ of an $n$-element set $S$.\\
{\bf {\em Question}:} Is there a 2-coloring of $S$, such that every 3
elements of $S$ that are consecutive according to one of the
permutations are not monochromatic?
\end{theorem}
\begin{proof}
We show a reduction from NAE 3SAT (not-all-equal 3SAT) which is the
following $NP$-complete problem~\cite{GareyJohnson-76}:\\
{\bf {\em  Input}:} A 3-CNF Boolean formula $\Phi$.\\
{\bf {\em  Question}:} Is there a $NAE$ assignment to $\Phi$? An
assignment is called $NAE$ if every clause has at least one
literal assigned True and at least one literal assigned False.\medskip

We first transform $\Phi$ into another instance $\Phi'$ in which
all variables are non-negated i.e., we make the instance monotone; this part of the proof is
folklore (see e.g., \cite{DBLP:conf/stoc/Schaefer78} for similar techniques). We then show how to realize $\Phi'$ using three permutations $\pi _1, \pi _2, \pi _3 $.

To transform $\Phi$ into $\Phi'$, for each variable $x$, we first replace the $i$th
occurrence of $x$ in its positive form by a variable
$x_i$, and the $i$th occurrence of $x$ in its negative form by
$x'_i$. The index $i$ varies between 1 and the number of
occurrences of each form (the maximum of the two). We also add the
following {\em consistency-clauses}, for each variable $x$ and for all $i$:
\begin{eqnarray*}
 \left(Z^x_i,x_i,x'_i      \right),
 \left(x_i,x'_i,Z^x_{i+1}  \right),
 \left(x_i,Z'^x_{i},x'_i \right)&,&
 \left(Z^x_{i},Z'^x_{i},Z^x_{i+1} \right)\\
 \left(x'_i, Z^x_{i+1}, x_{i+1} \right),
 \left(Z'^x_{i}, x'_i, x_{i+1} \right),
 \left(x'_i, x_{i+1}, Z'^x_{i+1},  \right)&,&
 \left(Z'^x_{i},Z^x_{i+1},Z'^x_{i+1} \right)
\end{eqnarray*}
where $Z_i^x$ and $Z'^x_i$ are new variables. This completes the
construction of $\Phi'$. Note that $\Phi'$ is monotone, as every negated variable has been replaced.

Moreover, $\Phi'$ has a $NAE$ assignment if and only if $\Phi$ has
a $NAE$ assignment. To see this, note that a $NAE$ assignment for
$\Phi$ can be translated to a $NAE$ assignment to $\Phi'$ as
follows: for every variable $x$ of $\Phi$ and every $i$, set
$x_i\equiv x$,
\hskip 0.07in $x'_i \equiv \overline{x}$,
\hskip 0.07in $Z_i^x \equiv True,
\hskip 0.07in Z'^x_i \equiv False$.

On the other hand, if $\Phi'$ has a $NAE$ assignment, then, by the
consistency clauses,
the variables in $\Phi'$ corresponding to any variable $x$ of $\Phi$ are assigned a consistent value.
Namely, for every $i,j$ we have $x_i = x_j$ and $x_i \neq
x'_i$. This assignment naturally translates to a $NAE$ assignment
for $\Phi$, by setting $x \equiv x_1$.\medskip

We next show how to realize $\Phi'$ by a set $S$ and three
permutations $\pi_1,\pi_2,\pi_3$. The elements of the set $S$ are
the variables of $\Phi'$, together with some additional elements
that are described below. Permutation $\pi_1$ realizes the clauses of $\Phi'$
corresponding to the original clauses of $\Phi$, while $\pi_2$ and
$\pi_3$ realize the consistency clauses of $\Phi'$.

The additional elements in $S$ are clause elements (two elements $c_{2j-1}$ and $c_{2j}$ for every clause $j$
of $\Phi$)
and \emph{dummy elements} $\star$ (the dummy elements are not indexed for the ease of presentation, but they appear in the same order in all three permutations).

Permutation $\pi_1$ encodes the clauses of $\Phi'$
corresponding to original clauses of $\Phi$ as follows (note that
all these clauses involve different variables). For each such
clause $(u,v,w)$, permutation $\pi_1$ contains the following
sequence: $$ c_{2j{-}1}, u, v, w, c_{2j}, \star, \star $$

\noindent At the end of $\pi_1$, for every variable $x$ of $\Phi'$
we have the sequence:
$$ Z^x_1, Z'^x_1, Z^x_2, Z'^x_2\
Z^x_3, Z'^x_3, \ldots, \star, \star$$

\noindent Permutation $\pi_2$ contains, for every variable $x$ of $\Phi$, the sequences:\\
$$Z^x_1, x_1, x'_1, Z^x_2, x_2, x'_2, Z^x_3, x_3, x'_3, Z^x_4, \ldots \hskip 0.5cm \text{and} \hskip 0.5cm \star, \star, Z'^x_1,
\star, \star, Z'^x_2, \star, \star, Z'^x_3, \ldots, \star, \star$$
At the
end of $\pi_2$ we have the clause-elements and
remaining dummy elements:\\
\noindent $$ \hskip 0.5cm \star, \star, c_{1}, \star, \star,
c_{2}, \star, \star, c_{3}\ldots$$

\noindent Similarly, permutation $\pi_3$ contains, for every variable $x$ of $\Phi$, the sequences:\\
$$x_1, Z'^x_1, x'_1, x_2, Z'^x_2, x'_2, x_3, Z'^x_3, x'_3, \ldots
\hskip 0.5cm  \text{and} \hskip 0.5cm  \star, \star, Z^x_1,
\star, \star, Z^x_2, \star, \star, Z^x_3, \ldots \star, \star$$ and at the end of $\pi_3$
we have the clause-elements and
remaining dummy elements:\\
 $$\star, \star, c_{1}, \star, \star, c_{2}, \star,
\star, c_{3},\ldots$$

This completes the construction of $S$ and $\pi_1,\pi_2,\pi_3$.
Note that for every clause of $\Phi'$ (whether it is derived from
$\Phi$ or is a consistency clause), the elements corresponding
to its three variables appear in sequence in one of the three
permutations. Therefore, if there is a 2-coloring of $S$, such that
every 3 elements of $S$ that are consecutive according to one of
the permutations are not monochromatic, then there is a $NAE$
assignment to $\Phi'$: each variable of $\Phi'$ is assigned True
if its corresponding element is colored `1', and False otherwise.

For the other direction, consider a $NAE$ assignment for $\Phi'$.
Observe that there is always a solution where $Z^x_i$ and $Z'^x_i$
are assigned opposite values. Then assign color `1' to elements
corresponding to variables assigned with True, and
assign color `0' to elements corresponding to variables
assigned with False. For the clause elements $c_{2j{-}1}$ and
$c_{2j}$ appearing in the subsequence $c_{2j{-}1},\ u,\ v,\ w,\
c_{2j}$, assign to $c_{2j{-}1}$ the color opposite to $u$, and to
$c_{2j}$ the color opposite to $w$. Finally, assign  colors `0' and `1' to each pair of
consecutive dummy elements, respectively.
It can be verified that there is no monochromatic
consecutive triple in any permutation.
\end{proof}

\paragraph{Approximation.}
Note that the general minimization problem (find a $k$-coloring that
minimizes the number of required points) can be approximated using the constructive version of the Lov\'asz Local Lemma (as mentioned in Section \ref{uperB}). Since $k$ is a trivial lower bound for any problem instance, this guarantees an approximation factor of $O( log k + log d)$. In particular, there exists a constant factor approximation for the problem introduced in Theorem \ref{theoNPhard} (since it fixes $d=3$ and $k=2$). 


\section{Coloring strips}
\label{sec_dual}

In this section we prove that any finite set of
strips in $\Re^d$ can be $k$-colored so that
every ``deep'' point is polychromatic. For a given set of strips (or intervals, if $d=1$),
we say that a point is {\em $i$-deep} if it is contained in at
least $i$ of the strips. We begin with the following easy lemma:

\begin{lemma}\label{intervals}
Let $\cal I$ be a finite set of intervals. Then
for every $k$, \hskip 0.02in $\cal I$ can be $k$-colored  so that
every $k$-deep point is polychromatic, while any point
covered by fewer than $k$ intervals will be covered by distinct colors.
\end{lemma}

\begin{proof}
We use induction on $\cardin{\cal I}$. Let $I$ be the interval
with the leftmost right endpoint. By induction, the intervals in
${\cal I} \setminus \{I\}$ can be $k$-colored with the desired
property. Sort the intervals intersecting $I$ according to their
left endpoints and let $I_1,\ldots,I_{k-1}$ be the first $k{-}1$
intervals in this order. It is easily seen that coloring $I$ with
a color distinct from the colors of those $k{-}1$ intervals
produces a coloring with the desired property, and hence a valid
coloring.
\end{proof}

\begin{theorem}
For any $d$ and $k$, one can $k$-color any set of axis-aligned strips
in $\Re^d$ so that every $d(k{-}1){+}1$-deep point
is polychromatic. That is,
$$\bar{p}(k,d)\le d(k{-}1)+1.$$
\end{theorem}

\begin{proof}
We start by coloring the strips parallel
to any given axis $x_i$ ($i=1,\ldots,d$) separately using the
coloring described in Lemma~\ref{intervals}. We claim that this
procedure produces a valid polychromatic coloring for all
$d(k{-}1){+}1$-deep points. Indeed assume that a given point $s$ is
$d(k{-}1){+}1$-deep and let $H(s)$ be the set of strips covering $s$.
Since there are $d$ possible orientations for the strips in
$H(s)$, by the pigeonhole principle at least $k$ of the strips
in $H(s)$ are parallel to the same axis. 
Then by property of the coloring of Lemma~\ref{intervals}, $H(s)$ is polychromatic.
\end{proof}

The above proof is constructive. By sorting the intervals that
correspond to any of the given directions, one can easily find a coloring in $O(n \log n)$ time.

We now give a lower bound on $\bar{p}(k,d)$. For that, we define $2d$ strips as follows: strip $s_{2i}$ is defined as $0< x_i <2$ (where $x_i$ is the coordinate of $i$-th dimension). Analogously, we define strip $s_{2i+1}$ as $1<x_i<3$. The main property of these strips is that we can always find a point covered by any subset of the $2d$ strips:
\begin{lemma}\label{lem_select}
For any $I \subseteq \{1,\ldots,2d\}$, there exists a point $p_I$ such that $p_I \in s_i \Leftrightarrow i \in I$, $\forall i \leq 2d$.
\end{lemma}
\begin{proof}
Note that whether or not point $p_I$ is covered by strips $s_{2i}$ or $s_{2i+1}$ only depends in the $i$-th coordinate of $p_I$. Thus, we define the $i$-th coordinate of point $p_I$ as follows:
\begin{itemize}
\item $1$ if $2i \in I$ but $2i +1 \not\in I$
\item $2$ if both $2i,2i+1 \in I$
\item $3$ if both $2i,2i +1 \not\in I$
\end{itemize}
Since the choice is independent on each dimension, the construction of $p_I$ is valid and is only covered by strips in $I$.
\end{proof}

We use these strips to find a lower bound on $\bar{p}(k,d)$ as follows:

\begin{theorem}\label{theo_lowbarp}
For any fixed dimension $d$ and integer $k$, it holds that
$$
\bar{p}(k,d) > \lfloor k/2 \rfloor d+1.
$$
\end{theorem}
\begin{proof}
Consider the $2d$ strips $\{s_i\}_{i \in [2d]}$ defined above.
We replace each strip $s_i$ with a cluster of $\lfloor k/2 \rfloor$
overlapping strips $\{s_{i,j}\}_{j \in [\lfloor k/2 \rfloor]}$, so that
a point $p_I$ defined in Lemma 3 is in strip $s_{i,j}$ if and only if
it is in $s_i$. This can be obtained, say, by perturbing a boundary of
the strip around $x_i=0$ (or around $x_i=3$).


Consider any coloring of the above problem instance with at most $\lfloor k/2 \rfloor d$ colors. As in the proof of Theorem \ref{thm_lb}, 
we can use the pigeonhole principle and the handshake lemma to conclude that there is at a color that is missing in at least $\lceil k/2 \rceil 2d/k \geq d$ clusters. Let $I$ be the set of at least $d$ indices whose clusters are missing the same color. By Lemma \ref{lem_select}, point $p_I$ is covered only by the strips in clusters $C_i$, for all $i \in I$. In particular, $p_I$ is at least $\lfloor k/2 \rfloor d$-deep and not colorful.
\end{proof}

\section*{acknowledgements}
This research was initiated during the WAFOL'09 workshop at Universit\'e Libre de {Bruxelles}~(U.L.B.), Brussels, Belgium.
The authors want to thank all other participants, and in particular Erik D. Demaine, for helpful discussions. Also, the authors would like to thank the anonymous referees for pointing out useful references.

\bibliographystyle{spmpsci}      
\bibliography{axis_parallel_strips}   

%
%

\end{document}